\documentclass[a4paper,11pt,english]{article}
\usepackage[utf8]{inputenc}
\usepackage{babel,amsmath,amssymb,amsthm,enumitem}
\usepackage[sort]{natbib}
\usepackage[colorlinks,allcolors=blue]{hyperref}
\usepackage[margin=1.25in]{geometry}
\usepackage[small,bf]{titlesec}
\titlelabel{\thetitle.\hspace{0.5em}}
\PassOptionsToPackage{usenames,dvipsnames}{xcolor}
\usepackage{tikz}

\newtheorem{theorem}{Theorem}
\newtheorem{lemma}{Lemma}
\newtheorem{corollary}{Corollary}
\newtheorem{proposition}{Proposition}
\theoremstyle{definition}
\newtheorem{remark}{Remark}
\newtheorem{definition}{Definition}

\newcommand{\floor}[1]{\lfloor{#1}\rfloor}

\DeclareMathOperator{\E}{E}
\DeclareMathOperator{\Var}{Var}

\DeclareMathOperator{\sgn}{sgn}
\DeclareMathOperator{\Law}{Law}
\renewcommand{\hat}{\widehat}
\renewcommand{\tilde}{\widetilde}
\renewcommand{\epsilon}{\varepsilon}
\renewcommand{\P}{\mathrm{P}}
\newcommand{\F}{\mathcal{F}}
\newcommand{\R}{\mathbb{R}}
\newcommand{\bt}{\tilde{b}}
\renewcommand{\b}{\bar b}
\newcommand{\M}{\mathcal{M}}
\let\Ss\S
\renewcommand{\S}{\mathcal{S}}
\newcommand{\cadlag}{c\`adl\`ag}

\title{Asymptotically optimal strategies in a diffusion approximation of a repeated betting game}
\author{Mikhail Zhitlukhin\thanks{Steklov Mathematical Institute of the
Russian Academy of Sciences. 8 Gubkina St., Moscow, Russia. Email:
mikhailzh@mi-ras.ru.}}
\date{26 August 2021}

\begin{document}
\maketitle

\begin{abstract}
We construct a diffusion approximation of a repeated game in which agents
make bets on outcomes of i.i.d.\ random vectors and their strategies are
close to an asymptotically optimal strategy. This model can be interpreted
as trading in an asset market with short-lived assets. We obtain sufficient
conditions for a strategy to maintain a strictly positive share of total
wealth over the infinite time horizon. For the game with two players, we
find necessary and sufficient conditions for the wealth share process to be
transient or recurrent in this model, and also in its generalization with
Markovian regime switching.


\end{abstract}

\section{Introduction}
In the present paper, we consider a dynamic game-theoretic model in which
agents make bets on outcomes of random events or random variables.
Investigation of this model is motivated by applications in analysis of
asymptotic performance of investment strategies in a multi-agent financial
market. The main aim of the paper is to construct and study a
continuous-time approximation of the model which arises when all agents make
``almost optimal'' bets.

To facilitate the exposition, let us begin with an example. Let
$(\Omega,\F,\P)$ be a probability space and $\{A_i^1,\ldots,A_i^N\}$,
$i=1,2,\ldots\,$, a sequence of independent partitions of $\Omega$ into $N$
random events. Assume that the probabilities $\P(A_i^n)$ are equal for all
$i$. Suppose that at each time $i\ge 0$, an agent bets a proportion
$\lambda_n\in[0,1]$ of her capital on the occurrence of $A_{i+1}^n$, where
$\lambda_1+\ldots+\lambda_N = 1$ (for further analysis it is important that
the whole capital is bet). At time $i+1$ the pool is divided between the
winning bets proportionally to their sizes, changing the distribution of
capital between the agents. We are interested in determining who of the
agents will have more wealth in the long run. In this paper, we consider
only fixed-mix (constant) strategies which are given exogenously; in
particular, they need not to form a Nash equilibrium.

In the above example, it is known that as $i\to\infty$ the entire market
wealth will be held by the agent whose strategy has the smallest
Kullback--Leibler divergence $D(p\,\|\,\lambda) = \sum_n p_n
\ln(p_n/\lambda_n)$ from the distribution $p=(p_1,\ldots,p_N)$, $p_n =
\P(A_i^n)$ \citep{BlumeEasley92}. This model has a natural interpretation of
an asset market consisting of $N$ Arrow securities (instruments with unit
payoffs in only one random state) with endogenous prices determined from
one-period equilibrium of variable asset demand, which depends on agents'
strategies, and fixed asset supply, see Remark~\ref{rem-prices}. Many
generalizations and extensions of this model have been obtained in the
literature. Among works in this direction which are closely related to the
material of the present paper, let us mention the paper of
\cite{EvstigneevHens+02}, in which Arrow securities are replaced with assets
that make random i.i.d.\ payoffs simultaneously. Further extensions were
obtained by \cite{AmirEvstigneev+05} who proved similar results for a model
with Markov payoffs, and \cite{AmirEvstigneev+13} who considered general
payoff sequences. The main results of the mentioned papers consist in
proving the existence of an ``unbeatable'' strategy which allows an agent to
survive in the market in the sense of maintaining a share of the total
market wealth strictly bounded away from zero over the infinite time
horizon. Under some additional conditions, such an agent turns out to be a
single survivor and accumulates in the limit the entire market wealth.

If no agent uses this optimal strategy, there might be several survivors,
even if some agent uses a strategy which is strictly closer to the optimal
strategy than the strategies of the other agents. This possibility depends
in an essential way on whether a market is complete or incomplete. In a
complete market there is always a single survivor, except some uninteresting
cases (the above example with independent partitions is a complete market).
A simple example of coexistence of surviving agents in an incomplete market
is provided by \citet[Ch.~9.3.3]{EvstigneevHens+09}. Let us also mention the
works of \cite{BottazziDindo14,BottazziGiachini17,BottazziGiachini19} who
studied survival and coexistence in a related setting but with agents'
strategies depending on one-step equilibrium asset prices. There is also a
large number of results on selection of agents by market forces in the
framework of general equilibrium, see, for example,
\cite{Sandroni00,BlumeEasley06} and references therein.

In the present paper we are interested in conditions for survival of agents
with fixed-mix strategies in a general (incomplete) market model, and focus
on the situation when strategies of agents are close to an optimal strategy.
The closeness is understood in the sense that we consider series of models
in which agents' strategies converge to the optimal strategy. This allows to
approximate the dynamics of the model by a system of stochastic differential
equations and investigate the solution of this system. From the point of
view of economic modeling, such an approximation is reasonable, since in the
long run we can leave out agents who make ``less correct'' predictions as
their share in the market wealth and influence on the dynamics of the model
will diminish with time. Although we do not obtain formal mathematical
conditions when an agent can be left out from the model, let us mention that
this idea is known in economics since long ago, see, e.g., \cite{Alchian50}
(however later studies show that it is not always applicable, see, e.g.,
\cite{DeLongShleifer+90,BlumeEasley06}).

Analytically, our approximation has an advantage over the pre-limit
discrete-time models, since it is easier to work with an SDE rather than a
recursive sequence defining the dynamics in discrete time. In particular,
this approximation becomes especially convenient in the case of two agents
and allows to thoroughly analyze the asymptotic behavior of the wealth
process.

The main results of the paper are as follows. First we prove the convergence
of the discrete-time model to the continuous-time model driven by a system
of SDEs. Then we obtain sufficient conditions for an agent to dominate or
survive in the continuous-time model. By survival we mean that the limit
superior of her share of total market wealth is strictly positive with
probability 1 as time goes to infinity. By dominance we mean that the limit
of the share of wealth is 1, i.e.\ this agent is a single survivor. These
conditions are obtained for the model with arbitrary number of agents. When
there are only two agents, we can go further and provide necessary and
sufficient conditions for survival and dominance, and, in the case when both
of the agents survive, show that the process of the share of wealth is
recurrent, determine when it is null or positive recurrent and find the
ergodic distribution. The latter result has a tight link with the stochastic
replicator equation of \cite{FudenbergHarris92}.

The paper is organized as follows. In Section~\ref{sec-discrete}, we
describe the discrete-time model and recall the main results known for it in
the literature. In Section \ref{sec-approximation}, we consider series of
discrete-time models and pass to the limit obtaining a continuous-time model
driven by a system of stochastic differential equations.
Section~\ref{sec-main-results} contains the main results about asymptotic
performance of agents' strategies. We first consider the case of many
agents, and then refine the obtained results in the case of two agents.
Illustrations and numerical examples are provided in
Section~\ref{sec-examples}. In Section~\ref{sec-switching}, we study an
extension of the two-agent case in which the market is modeled by the same
SDE but with switching between two regimes. The \ref{appendix} contains a
theorem on convergence in distribution of a discrete-time sequence to a
diffusion process in a form convenient for our purposes.

\section{A discrete-time model}
\label{sec-discrete}
Let $(\Omega,\F,\P)$ be a probability space on which all random variables
will be defined. Equalities and inequalities for random variables will be
understood to hold with probability~1, unless else is stated.

There are $M\ge 2$ agents and $N\ge 2$ assets in the model. The time is
discrete, $i=0,1,2,\ldots$ At each moment of time $i\ge 1$, the assets yield
random payoffs $X_i^n$, $n=1,\ldots,N$, which are divided between the agents
proportionally to the amount of wealth each agent invests in an asset. The
random vectors $X_i=(X_i^1,\ldots,X_i^N)$ are i.i.d.

The wealth of the agents is represented by random sequences $Y_i^m$,
$m=1,\ldots,M$, $i\ge 0$, which are defined inductively as follows. The
initial wealth $ Y_0^m$ is non-random and strictly positive. At each moment
of time, agent $m$ splits the available wealth for investing in the assets
in proportions $\lambda_m = (\lambda_{m1},\ldots,\lambda_{mN})$. The vector
$\lambda_m$ represents the strategy of this agent. The whole wealth is
reinvested, so $\lambda_m$ belongs to the standard $N$-simplex $\Delta_N
=\{\lambda\in \R_+^{N} : \sum_n \lambda^n =1\}$. We consider only constant
strategies which depend neither on time nor on a random outcome. Then the
wealth sequences $Y_i^m$ are defined by the equation (see
Remark~\ref{rem-prices} below for an interpretation)
\begin{equation}
Y_{i+1}^m =  \sum_{n=1}^N \frac{\lambda_{mn}  Y_{i}^m}{\sum_k
\lambda_{kn}  Y_{i}^k} X_{i+1}^n.\label{wealth-discrete}
\end{equation}
We will assume that there is at least one agent who allocates a strictly
positive proportion of wealth in every asset, i.e.\
\[
\lambda_{mn} > 0\ \text{for some $m$ and all $n$}.
\]
Under this conditions, $ Y_i^m>0$ for all $i$, the total market wealth
$\sum_m Y_i^m$ does not depend on the agents' strategies and is equal to
$\sum_n X_i^n$.

We will be interested in the asymptotic behavior of the relative wealth of
agents
\[
R_i^m = \frac{ Y_i^m}{\sum_{k}  Y_i^k}.
\]
It is easy to see that
\[
\frac{ R_{i+1}^m}{ R_i^m} = \sum_{n=1}^N\biggl( \frac{\lambda_{mn}}{\sum_k
\lambda_{kn}  R_i^k} \cdot \frac{X_i^n}{\sum_l X_i^l}\biggr).\label{R-initial}
\]
In particular, the relative wealth does not change under scaling of the
vector $X_i$ and therefore in what follows we will assume that
\[
\sum_{n=1}^N X_i^n =1\ \text{for all $i\ge 1$}, \qquad \sum_{m=1}^M Y_0^m = 1.
\]
Under this assumption we have $ R_i^m= Y_i^m$ and $\sum_m Y_i^m = 1$.

\begin{remark}
\label{rem-prices}
One can interpret the model defined by equation \eqref{wealth-discrete} as
an asset market in which at every moment of time agents can buy $N$ assets
which yield random payoffs at the next moment of time. If their prices
$P_i^n$ are determined from the equality of supply and demand, then agent
$m$ buys $x_i^{mn} = \lambda_{mn} Y_i^m/ P_i^n$ units of asset $n$. On the
other hand, the equality of supply and demand implies that $P_i^m = \sum_m
\lambda_{mn} Y_i^m / S_n$, where $S_n$ is the supply of asset $n$. Hence
agent $m$ will receive the payoff $x_i^{mn} X_{i+1}^n/S_n$, where
$X_{i+1}^n/S_n$ is the payoff per one unit of asset $n$. This gives equation
\eqref{wealth-discrete}.

It should be noted that this model assumes the assets are short-lived in the
sense that they are bought by the agents, yield payoffs at the next moment
of time, then expire and get replaced by new assets (so they live for just
one period). Such assets can used to model standardized contracts, for
example, derivative securities, agreements to produce or deliver goods or
services, etc. The model differs from a usual stock market model in
mathematical finance, see, e.g., \cite{EvstigneevHens+16} for details and
interpretations.

Observe that the model described in the introduction, which corresponds to
the case of a complete market, is obtained if
\begin{equation}
\P(X_i \in \{e_1,\ldots,e_N\}) = 1,\label{complete}
\end{equation}
where 
$e_i=(0,\ldots,1,\ldots,0)$ are the standard basis vectors.
\end{remark}

To motivate further discussion, let us state a result on asymptotically
optimal strategies in the model under consideration. We will say that there
are no redundant assets in the market if there is no non-trivial linear
combination $c_1 X_i^1 + \ldots+c_N X_i^N$ equal to a constant vector with
probability 1.

\begin{proposition}
\label{pr-known-results}
Suppose agent $m$ uses the strategy $\lambda_m = \hat \lambda := (E
X_i^1,\ldots, \E X_i^N)$. Then for any strategies of the other agents it
holds that (with probability 1)
\[
\inf_{i\ge 0} Y_i^m > 0.
\]
If agent $m$ uses a strategy different from $\hat \lambda$, then it is
possible to find strategies $\lambda_k$ of agents $k\neq m$ (one can take
$\lambda_k = \hat \lambda$) such that
\[
\lim_{i\to \infty} Y_i^m = 0.
\]
If there are no redundant assets and at least one agent uses the strategy
$\hat \lambda$, then for any agent $k$ who uses a different strategy it
holds that
\[
\lim_{i\to\infty} Y_i^k = 0.
\]
\end{proposition}
These statements were proved by \cite{EvstigneevHens+02}; see also
\cite{AmirEvstigneev+05} for similar results in a model with Markov payoff
sequences, and \cite{AmirEvstigneev+13} for a model with general payoffs.
Observe that in model \eqref{complete}, the strategy $\hat \lambda$ has the
components $\hat\lambda^n = \P(X_i^n=1)$, and therefore is sometimes called
the Kelly strategy after \cite{Kelly56} (the Kelly strategy consists in
``betting one's beliefs'', i.e.\ in proportion to probabilities of
outcomes).

Note that the formula for the optimal strategy $\hat\lambda$ appearing in
Proposition~\ref{pr-known-results} has a relatively simple form largely due
to the assumption that the whole capital is reinvested in the assets. If
agents are allowed to keep part of their wealth not invested (i.e.\ $\sum_n
\lambda_{mn} \le 1$), the optimal strategy will not be constant; see
\cite{DrokinZhitlukhin20,Zhitlukhin21b} for details; an extension to
continuous time can be found in \cite{Zhitlukhin21a,Zhitlukhin20}.

Proposition~\ref{pr-known-results} shows that the strategy $\hat\lambda$
drives other strategies out of the market in the long run, which can be
regarded as a form of asymptotic optimality. However this strategy requires
an agent to have precise estimates of the expected payoffs $\E X_i^n$, which
may be difficult to achieve. In view of that and as discussed in the
introduction, it becomes interesting to consider a model where agents have
``almost'' precise estimates and study it in the limit when the estimation
error vanishes. The next section describes such a model using an appropriate
approximation with diffusion processes.

\section{Approximation by a diffusion process}
\label{sec-approximation}
Consider series of the above discrete-time models indexed by a parameter
$\delta>0$ with asset payoffs and investors strategies satisfying the
relations
\begin{align}
&\E X_i^{\delta,n} = \mu_n + a_n\sqrt \delta + c_n(\delta), \label{EX-series}\\
&\lambda_{mn} = \mu_n + b_{mn}\sqrt \delta + d_{mn}(\delta) , \label{lambda-series}\\
&\mathrm{cov}(X_i^{\delta,n}, X_i^{\delta,l}) = \sigma_{nl} + e_{nl}(\delta), \label{cov-series}
\end{align}
where $\mu_n>0$, $\sum_n \mu_n = 1$, $\sum_n a_n = \sum_n b_{mn} = \sum_n
c_n(\delta) = \sum_n d_{mn}(\delta) = 0$, the matrix $\sigma =
(\sigma_{nl})\in \R^{N\times N}$ is symmetric and non-negative definite, the
functions $c_n(\delta)$, $d_{mn}(\delta)$, $e_{nl}(\delta)$ have the limits
$\lim_{\delta\downarrow0}c_n(\delta)/\sqrt\delta =
\lim_{\delta\downarrow0}d_{mn}(\delta)/\sqrt\delta =
\lim_{\delta\downarrow0} e_{nk}(\delta) = 0$.

Denote by $\tilde Y_i^{\delta} = (\tilde Y_i^{\delta,1},\ldots,\tilde
Y_i^{\delta,M})$ the wealth sequences of the agents (tildes will be used in
the notation to distinguish discrete-time objects). The vector of initial
wealth $\tilde Y_0$ is assumed to be the same for all $\delta$, with $\tilde
Y_0^m>0$ for all $m$. Denote by $Y_t^{\delta}$ the piecewise-constant
embedding of $\tilde Y_i^{\delta}$ into continuous time with step $\delta$,
i.e.\
\[
Y_t^{\delta} = \tilde Y^{\delta}_{\lfloor t/\delta \rfloor}.
\]

The next theorem contains the main result about the convergence of the
discrete-time models to a continuous-time model. Everywhere the convergence
will be understood as the weak convergence of distributions on the Skorokhod
space. Let $\bar b(y)\colon \Delta_M \to \Delta_N$ denote the weighted
coefficient $b$ of the strategies of the agents with a vector of weights
$y$, i.e.
\[
\b_n(y) = \sum_{m=1}^M y^m b_{mn}.
\]
Consider the system of $M$ stochastic differential equations
($m=1,\ldots,M$)
\begin{equation}
d Y^m_t = Y_t^m\sum_{n=1}^N \frac{1}{\mu_n}\biggl( (b_{mn} -
\b_n(Y_t))(a_n - \b_n(Y_t)) dt + (b_{mn} - \b_n(Y_t)) d W^n_t\biggr),
\label{sde}
\end{equation}
where $W_t^n$ are correlated Brownian motions with zero mean and covariance
\begin{equation}
\E(W_t^n W_t^l) = \sigma_{nl} t.\label{cov-B}
\end{equation}

\begin{theorem}
\label{th-limit}
Equation \eqref{sde} has a unique strong solution $Y$ for any initial
condition $Y_0 \in \Delta_N$ and $\Law(Y^\delta_t,\ t\ge 0) \to \Law(Y_t,\
t\ge0)$ as $\delta\to 0$.
\end{theorem}

Before giving a proof, which will be based on verification of some technical
conditions, let us provide an intuitively clear (but not formally rigor)
argument explaining why one can expect equation \eqref{sde} in the limit.
Expanding \eqref{wealth-discrete} in the Taylor series up to order $\delta$,
we get
\[
\begin{split}
\tilde Y_{i+1}^{\delta,m} &\approx
\tilde Y_i^{\delta,m}\biggl\{ 1 +
\sqrt\delta \sum_{n=1}^N \frac1{\mu_n}(b_{mn} - \b_n(\tilde Y_i^{\delta}))
  (X_{i+1}^{\delta,n} - \E X_{i+1}^{\delta,n}) \\
&\qquad+ \delta \sum_{n=1}^N \frac1{\mu_n} (b_{mn} -\b_n(\tilde Y_i^{\delta}))
  (a_n - \b_n(\tilde Y_i^{\delta})) \\&\qquad+
  \sum_{n=1}^N \biggl(\frac{d_{mn}(\delta)}{\mu_n} +
    \frac{\delta}{\mu_n^2}\b_n(\tilde Y_i^{\delta})(\b_n(\tilde Y_i^{\delta}) - b_{mn})\biggr)
    (X_{i+1}^{\delta,n} - \E X_{i+1}^{\delta,n}) \biggr\}\\ &:=
\tilde Y_i^{\delta,m}\{1 + \sqrt\delta A_{i+1}^{\delta} + \delta B_{i+1}^{\delta} +
 C_{i+1}^{\delta} \}.
\end{split}
\]
By Chebyshev's inequality, one can see that for any $t>0$ we have
$\sum_{i<\lfloor t/\delta \rfloor} \tilde Y_i^{\delta,m}C_{i+1}^\delta \to
0$ in probability as $\delta\to 0$. Hence for the embedding in continuous
time we get
\begin{equation}
\begin{split}
Y_t^{\delta,m} &= \tilde Y_{\lfloor t/\delta\rfloor}^{\delta,m} \approx 
\delta \sum_{i<\lfloor t/\delta\rfloor} \tilde Y_i^{\delta,m}
  \sum_{n=1}^N \frac1{\mu_n} (b_{mn} -\b_n( \tilde Y_i^{\delta}))
    (a_n - \b_n(\tilde Y_i^{\delta})) \\ &+
\sqrt\delta\sum_{i<[t/\delta]}\tilde Y_i^{\delta,m}
  \sum_{n=1}^N \frac{1}{\mu_n}(b_{mn} - \b_n(\tilde Y_i^{\delta})) (X_{i+1}^{\delta,n} - \E X_{i+1}^{\delta,n}) 
:= \delta F_t^\delta + \sqrt\delta G_t^\delta.
\end{split}
\label{FG}
\end{equation}
Suppose there exists the limit process $Y_t = \lim_{\delta\to 0}
Y_t^\delta$. Then we can approximate $\tilde Y^{\delta,m}_i$ with
$Y^m_{\delta i}$ in \eqref{FG}, which implies
\[
\delta F_t^\delta \approx \sum_{n=1}^N \int_0^t  \frac{Y^m_s}{\mu_n} (b_{mn} -
\b_n(Y_s))(a_n - \b_n(Y_s)) ds.
\]
As for the term $\sqrt\delta G_t^\delta$, since $X_i^\delta$ is a sequence
of i.i.d.\ vectors, we can approximate it with a sequence of increments of a
multidimensional Brownian motion with appropriate correlation matrix.
Namely,
\[
\sqrt\delta G_t^\delta\approx
\sum_{n=1}^N \int_0^t  \frac{Y^m_s}{\mu_n} (b_{mn} - \b_n(Y_s)) d W^n_s,
\]
where $W_t^n$ are Brownian motions, which, in view of \eqref{cov-series},
satisfy \eqref{cov-B}.

\begin{proof}[Proof of Theorem~\ref{th-limit}]
It is easy to see that for any initial condition $Y_0\in \Delta_M$, a
solution of \eqref{sde} must always stay in $\Delta_M$. Indeed, $Y_t$ is
clearly non-negative. The equality $\sum_m Y_t^m = 1$ follows from that
$\sum_m y^m (b_{mn} - \b_n(y)) = \b_n(y)(1-\bar y)$, where $\bar y = \sum_m
y^m$, and hence
\[
d \bar Y_t = \sum_{n=1}^N \frac1{\mu_n}{\b_n(Y_t)(1-\bar Y_t)}(  (a_n-\b_n(Y_t))  dt +
 dB_t^n).
\]
So, if $\bar Y_0 = 1$, then $\bar Y_t = 1$ for all $t\ge 0$.

Consequently, without loss of generality, the coefficients of \eqref{sde}
can be modified outside of $\Delta_M$ and replaced with functions $f(y)$ and
$g(y)$ (the drift and diffusion coefficients, respectively) which are smooth
on $\R^M$, have a bounded support, and for $y\in\Delta_M$ \[ f^m(y) =
y^m\sum_{n=1}^N \frac1{\mu_n} (b_{mn} - \b_n(y))(a_n - \b_n(y)),\qquad
g^m(y) = y^m\sum_{n=1}^N\frac1{\mu_n}(b_{mn} - \b_n(y)).
\]
The existence and uniqueness of a strong solution of \eqref{sde} follows
from classical Ito's theorem. To prove the convergence of distributions, we
will apply Proposition~\ref{th-convergence} from the \ref{appendix}.
Conditions \ref{conv-1}, \ref{conv-2} of this proposition are clearly met.
Condition \ref{conv-3} holds with the function $F(t) = t \max_y (\sum_m
|f^m(y)| + \mathrm{tr}(g(y)\sigma g(y)^T))$.

Let us check \ref{conv-4}. Since $\tilde Y_i^\delta = Y^\delta_{i\delta}$ is
a homogeneous Markov sequence, we can find functions $f^\delta(y)\colon
\R^M\to \R^M$ such that
\[
\E(\Delta Y_{i\delta}^{\delta} \mid \F_{(i-1)\delta}^\delta)
=  f^\delta(Y^\delta_{(i-1)\delta}),
\]
where $\F_{t}^\delta = \sigma(Y_s^\delta;\; s\le t)$. For $\alpha \in
D(\R_+; \Delta_M)$, define
\[
B_t^\delta(\alpha) =  \sum_{1\le i \le \floor{t/\delta}}
 f^\delta(\alpha^\delta_{(i-1)\delta}), \qquad
B_t(\alpha) = \int_0^t f(\alpha_s) ds,
\]
so that $B_t^\delta(Y^\delta)$ and $B_t(Y)$ are the first predictable
characteristics of the processes $Y_t^\delta$ and $Y_t$ (see
\eqref{B-delta}, \eqref{B-C}). It will be enough to show that for any
sequence of functions $\alpha^\delta \in D(\R_+; \Delta_M)$ which are
piecewise-constant on intervals $[i\delta,(i+1)\delta)$ and any $t\ge0$ we
have
\begin{equation}
\lim_{\delta\to 0}\sup_{s\le t} \|B_s^\delta(\alpha^\delta) - B_s(\alpha^\delta)\|= 0.\label{b-conv}
\end{equation}
To compute $f^\delta$, observe that \eqref{wealth-discrete} and
\eqref{lambda-series} imply (with $\bar d_n(y) = \sum_m d_{mn} y^m$)
\begin{equation}
\begin{split}
\Delta Y_{i\delta}^{\delta,m} &=
Y_{(i-1)\delta}^{\delta,m} \biggl(\sum_{n=1}^N
  \frac{\mu_n+b_{mn}\sqrt \delta + d_{mn}(\delta)}{\mu_n + \b_{n}(Y_{(i-1)\delta}^{\delta})
\sqrt\delta + \bar d_n(Y_{(i-1)\delta}^{\delta})(\delta)} X_{i}^{\delta,n} -
1\biggr) \\ &=
Y_{(i-1)\delta}^{\delta,m}\sum_{n=1}^N \biggl( \frac{\sqrt\delta}{\mu_n}(b_{mn} - \b_n(Y_{(i-1)\delta}^\delta)) +
  \frac{\delta}{\mu_n^2}(\b_n^2(Y_{(i-1)\delta}^\delta) - b_{mn}\b_n(Y_{(i-1)\delta}^\delta)) \\
&\phantom{=}+\frac{1}{\mu_n}( d_{mn}(\delta) - \bar d_n(Y_{(i-1)
  \delta}^\delta)(\delta))\biggr)
  X_i^{\delta,n} + \delta \rho^m(\delta; i),
\end{split}
\label{delta-y}
\end{equation}
where $\rho(\delta; i) =\rho(\omega;\delta;i)$ is a family of vectors in
$\R^M$ with $\mathrm{ess\,sup}_\omega \sup_{i\ge 1} \|\rho(\delta;i)\| \to
0$ as $\delta\to0$. The above equation was obtained by expanding the
denominator in Taylor series and using the relation $\sum_n
X_i^{\delta,n}=1$. Consequently, using \eqref{EX-series} we find
\[
f^{\delta,m}(y) = 
\delta y^m\sum_{n=1}^N \frac1{\mu_n}(b_{mn} - \b_n(y))(a_n- \b_n(y)) + \delta \rho^m(\delta;y), 
\]
where $\rho(\delta;y)$ is a function with values in $\R^M$ such that
$\rho(\delta;y)\to 0$ uniformly in $y\in\Delta_M$ as $\delta \to 0$. Here we
used the relation $\sum_n b_{mn} = \sum_n \b_n(y) = \sum_n d_{mn} = \sum_n
\bar d_n(y)=0$. Thus, for any $t=j\delta $, we have
\[
\|B_t^\delta(\alpha^\delta) - B_t(\alpha^\delta)\| \le
\delta \sum_{1\le i \le j} \|\rho(\delta; \alpha^\delta_{(i-1)\delta})\|,
\]
which implies \eqref{b-conv}.

In a similar way, we can verify condition \ref{conv-5}. Let
$g^{\delta}\colon \R^M \to \R^{M\times M}$ be defined by
\[
\begin{split}
g^{\delta,mk}(y) &= 
\E(\Delta Y_{i\delta}^{\delta,m} \Delta Y_{i\delta}^{\delta,k} \mid 
Y_{(i-1)\delta}^\delta = y) \\&\quad- \E(\Delta Y_{i\delta}^{\delta,m} \mid Y_{(i-1)\delta}^\delta = y)
\E(\Delta Y_{i\delta}^{\delta,k} \mid Y_{(i-1)\delta}^\delta = y) \\ &=
\delta y^my^k \sum_{n,l=1}^N \frac{\sigma_{nl}}{\mu_n\mu_l}(b_{mn} - \b_n(y))(b_{kl} - \b_l(y))
 + \delta \rho^{mk}(\delta;y),
\end{split}
\]
where $\rho(\delta;y)$ is (another) function converging to 0 uniformly in
$y\in\Delta_M$ as $\delta\to0$. The second modified predictable
characteristics of $Y_t^\delta$ and $Y_t$ are, respectively,
$C_t^\delta(Y^\delta)$ and $C_t(Y)$, where
\[
C_t^\delta(\alpha) = \sum_{1\le i \le \floor{t/\delta}}
 g^\delta(\alpha^\delta_{(i-1)\delta}), \qquad
C_t(\alpha) = \int_0^t
g(\alpha_s)\sigma g(\alpha_s)^T\, ds.
\]
Hence for $t=j\delta$ and $\alpha^\delta \in D(\R_+; \Delta_M)$ which are
piecewise-constant on $[i\delta,(i+1)\delta)$ we have
\[
\|C_t^\delta(\alpha^\delta) - C_t(\alpha^\delta)\| \le
\delta\sum_{1\le i \le j}  \|\rho(\delta, \alpha^\delta_{(i-1)\delta})\| \to 0,
\]
which gives \ref{conv-5}.

Finally, condition \ref{conv-6} holds because if $h(y)$ is a function
vanishing in a neighborhood of zero, then for sufficiently small $\delta$
all the jumps $\Delta Y_{i\delta}^\delta$, $i\ge 0$, lie in such a
neighborhood with probability 1, see \eqref{delta-y}.
\end{proof}

\section{Asymptotic relative performance of strategies}
\label{sec-main-results}
\subsection{General results for an arbitrary number of agents}
In the rest of the paper we will work within the continuous-time model
obtained in the previous section and identify agents' strategies with
vectors $b_m=(b_{m1},\ldots,b_{mN})^T$ from \eqref{lambda-series}. Denote
also $a=(a_1,\ldots,a_N)^T$, where $a_n$ are the coefficient from
\eqref{EX-series}. We will be primarily interested in the relative
performance of strategies as $t\to\infty$.

\begin{definition}
We shall say that in a strategy profile $b=(b_1,\ldots, b_M)$ with initial
wealth $Y_0 = (Y_0^1,\ldots,Y_0^M)$ agent $1$
\begin{itemize}[leftmargin=*,label=$-$,itemsep=0.0em]
\item vanishes if $\lim\limits_{t\to\infty} Y_t^1 = 0$ a.s.;
\item survives\footnote{In some works survival means $\liminf\limits_{t\ge0}
Y_t^1 > 0$.} if $\limsup\limits_{t\to\infty} Y_t^1 > 0$ a.s.;
\item dominates if $\lim\limits_{t\to\infty} Y_t^1  = 1$ a.s.
\end{itemize}
\end{definition}

The next theorem provides sufficient conditions for an agent to survive,
dominate, or vanish. For brevity of notation, introduce the matrices
\begin{equation}
\M = \mathrm{diag}\biggl(\frac1{\mu_1},\ldots,\frac1{\mu_N}\biggr), \qquad \S =
\biggl(\frac{\sigma_{nl}}{\mu_n\mu_l}\biggr)_{n,l=1}^N.\label{matrices}
\end{equation}

\begin{theorem}
\label{th-many-agents}
Fix a strategy profile $(b_1,\ldots, b_M)$ and a vector of initial wealth
$(Y_0^1,\ldots,Y_0^M)$ with $Y_0^1>0$. Let $B =
\mathrm{conv}(b_2,\ldots,b_M)$ denote the convex hull of the strategies of
agents $m\ge 2$. Define the coefficients
\begin{align}
&\theta_0 = \inf_{\tilde b \in B}\Bigl( (a-b_1)^T\M(b_1-\tilde b) + \frac12(b_1-\tilde b)^T
(2\M-\S)(b_1-\tilde b)\Bigr),\label{d0}\\
&\theta_1 = \inf_{\tilde b \in B}\Bigl( (a-b_1)^T\M(b_1-\tilde b) + \frac12
(b_1-\tilde b)^T \S(b_1-\tilde b) \Bigr).\label{d1} 
\end{align}
Then $\theta_0\ge \theta_1$ and the following statements are true:
\begin{enumerate}[leftmargin=*,label=(\alph*)]
\item\label{th-many-agents-1} if $\theta_0 > 0$ or $\theta_1\ge 0$, then agent
$1$ survives;
\item\label{th-many-agents-2} if $\theta_1>0$, then agent $1$ dominates;
\item\label{th-many-agents-3} if $(a-b_1)^T\M(b_1-b_m) \ge 0$ for
$m=2,\ldots,M$, then there exists $\lim_{t\to\infty} Y_t^1 > 0$ a.s., and,
in particular, agent 1 survives.
\end{enumerate}
\end{theorem}
\begin{remark}
As will be seen from the proof of the theorem, the coefficients $\theta_0$
and $\theta_1$ give lower bounds for the drift coefficient of the process
$\ln (Y_t^1/(1-Y_t^1))$ when $Y_t^1$ is close to 0 or 1, respectively. See
\eqref{Zt} below.

Observe that we have $\theta_0=\theta_1$ if the matrix $\S$ has the form
\begin{equation}
\label{cont-complete}
\S_{nl} = -1\ \text{for}\ n\neq l, \qquad \S_{nn} = \frac{1}{\mu_n}-1,
\end{equation}
because in this case $\M-\S$ is the matrix of all units, so
$(b_1-\bt)^T(\M-\S)(b_1-\bt) = (\sum_n (b_{1n}-\bt_n))^2 = 0$, which implies
$\theta_0=\theta_1$. In particular, \eqref{cont-complete} takes place if the
pre-limit discrete-time models represent a complete market in the sense of
\eqref{complete}.
\end{remark}

In order to prove Theorem~\ref{th-many-agents}, we will first establish the
following auxiliary inequality of a general nature.

\begin{lemma}
\label{lem-inequality}
Let $X$ be a random vector in $\Delta_N$, $\mu_n = \E X_n>0$, $\sigma_{nl} =
\mathrm{cov}(X_n,X_l)$, and define $\M,\S$ as in \eqref{matrices}. Then for
any $c\in \R^N$ such that $\sum_{n} c_n = 0$ we have
\begin{equation}
c^T(\M-\S)c \ge 0.\label{m-s-ineq}
\end{equation}
\end{lemma}
\begin{proof}
Since any distribution can be approximated by a discrete one, it is
sufficient to prove the lemma in the case when $X$ has a discrete
distribution.

Fix a set $\{x_1,\ldots,x_K\} \subset \Delta_N$. We are going to show that
inequality \eqref{m-s-ineq} holds true for any distribution
$p=(p_1,\ldots,p_K)$, $p_k = \P(X=x_k)$, such that $\E X_n=\mu_n$,
$\mathrm{cov}(X_n,X_l)=\sigma_{nl}$. Observe that
\[
c^T(\M-\S)c = c^T\M c - \E(c^T\M X)^2 = c^T\M c - \sum_{k=1}^K (c^T\M x_k)^2 p_k.
\]
Consider the following linear programming problem with variables
$p_1,\ldots,p_K$:
\begin{align}
\text{minimize}\ &\ v(p):=c^T\M c - \sum_{k=1}^K (c^T\M x_k)^2 p_k\notag\\
\text{subject to}\ &
\sum_{k=1}^K x_{kn} p_k = \mu_n, \quad n=1,\ldots,N,\label{min-prob-1}\\
&\sum_{k=1}^K p_k = 1,\label{min-prob-2}\\
&p_k \ge 0, \quad k=1,\ldots,K.\label{min-prob-3}
\end{align}
Since the constraint set is non-empty and compact, the minimizer $p^*$
exists. We need to show that $v(p^*)\ge 0$. To that end, consider the dual
problem with variables $q = (q_1,\ldots,q_{N+1})$ which correspond to
equality constraints \eqref{min-prob-1}--\eqref{min-prob-2} (see, e.g.,
\citet[Ch.~5.2.1]{BoydVandenberghe04}):
\begin{align}
\text{maximize}\ &\ d(q):=c^T\M c -\sum_{n=1}^N \mu_n q_n - q_{N+1}\notag \\
\text{subject to}\ &
\sum_{n=1}^N x_{kn}q_n + q_{N+1} - (c^T\M x_k)^2 \ge 0, \qquad k=1,\ldots,K.
\label{dual-prob-constr}
\end{align}
Let $q = (q_1,\ldots,q_{N+1})$ be defined by
\[
q_n = \frac{c_n^2}{\mu_n^2}\ \text{for}\ n=1,\ldots,N,\quad q_{N+1} = 0.
\]
It is easy to check that $d(q)= 0$ and $q$ satisfies constraints
\eqref{dual-prob-constr} (this follows from applying Jensen's inequality and
treating each vector $x_k\in\Delta_N$ as coefficients of a convex
combination). Then $v(p^*)\ge d(q) = 0$ in view of the duality.
\end{proof}

\begin{proof}[Proof of Theorem~\ref{th-many-agents}]
The inequality $\theta_0\ge \theta_1$ follows from
Lemma~\ref{lem-inequality} with $c=b_1-\bt$.

To prove claims \ref{th-many-agents-1} and \ref{th-many-agents-2}, let $Z_t
= \ln (Y_t^1/(1-Y_{t}^1))$ and denote by $\bt_t$ the weighted strategy of
agents $m\ge2$:
\[
\bt_t = \sum_{m=2}^M \frac{Y^m_t}{1-Y^1_t} b_m.
\]
By Ito's formula, we have
\begin{equation}
d Z_t = \gamma_t dt + (b_1-\bt_t)^T \M dW_t,\label{Zt}
\end{equation}
where
\[
\gamma_t =  (a-\bt_t)^T \M(b_1-\bt_t) - \frac12(b_1-\bt_t)^T\S(b_1-\bt_t) +
Y_t^1 (b_1-\bt_t)^T (\S-\M)(b_1-\bt_t).
\]
Notice that $\theta_0$ and $\theta_1$ give the minimum possible values for
$\gamma_t$ when the value of $Y_t^1$ approaches 0 and 1, respectively.

Suppose $\theta_0 > 0$. Then on the set $\Omega' = \{\lim_{t\to\infty} Y_t^1
= 0\} = \{\lim_{t\to\infty} Z_t = -\infty\}$ we have
\[
\lim_{t\to\infty}\frac{1}{t}\int_0^t \gamma_s ds >\frac{\theta_0}{2}>0,
\]
while by the strong law of large numbers for martingales we have with
probability 1
\[
\lim_{t\to\infty}\frac1t\int_0^t (b-\bt_s)^T \M d W_s = 0.
\]
Hence on $\Omega'$ we have $\lim_{t\to\infty} t^{-1}Z_t > 0$ a.s., which by
the definition of $\Omega'$ implies $\P(\Omega') = 0$, and proves
claim~\ref{th-many-agents-1} in the case $\theta_0>0$.
Claim~\ref{th-many-agents-2} is proved in a similar way, but using that
$\lim_{t\to\infty} t^{-1}\int_0^t \gamma_sds \ge \theta_1>0$ a.s.

To prove \ref{th-many-agents-1} when $\theta_0=\theta_1=0$, observe that in
this case $\gamma_t \ge 0$, so the process $Z_t$ is a submartingale, and a
continuous submartingale cannot have the limit $-\infty$ with positive
probability.\footnote{The process $Z_t^a = Z_{t\wedge \tau_a}$, where
$\tau_a = \inf\{t\ge 0: Z_t \ge a\}$ is a submartingale bounded from above,
hence it has a finite limit. Therefore, $\lim_{t\to\infty}Z_t$ exists and is
finite on the set $\{\sup_{t\ge0} Z_t <\infty\}$, while on the complementary
set we have $\limsup_{t\to\infty} Z_t = +\infty$.}

To prove claim \ref{th-many-agents-3}, consider the process $\ln Y_t^1$. By
Ito's formula,
\begin{multline*}
d \ln Y_t^1 = (1-Y_t^1)\Bigl\{\Bigl( (a-b_1)^T\M(b_1-\bt_t) \\+ \frac12(1-Y_t^1)
(b_1-\bt_t)^T(2\M-\S)(b_1-\bt_t)\Bigr) dt +
(b_1-\bt_t)^T\M dW_t\Bigr\}.
\end{multline*}
If the condition of the claim holds true, then Lemma~\ref{lem-inequality}
implies that the drift coefficient is non-negative, so $\ln Y_t^1$ is a
submartingale. Since it is bounded from above, there limit
$\lim_{t\to\infty} \ln Y_t^1$ is finite, hence $\lim_{t\to\infty} Y_t^1 >
0$.
\end{proof}

\subsection{The case of two agents}
\label{sec-two-agents}
When $M=2$, it is possible to give a more thorough characterization of the
asymptotics of the agents' wealth by using standard results on ergodicity of
Markov processes (see, e.g., \citet[\Ss\,16,18]{GikhmanSkorokhod72} for
details on results that are needed below). In this case the wealth dynamics
is determined by the one-dimensional equation for the wealth of one agent
and we have
\begin{equation}
dY_t^1 = Y_t^1(1-Y_t^1)\Bigl\{\bigl((a- b_2)^T\M(b_1- b_2) -Y_t^1(b_1-
b_2)^T\M(b_1- b_2)\bigr) dt + v dW_t\Bigr\},\label{sde-2}
\end{equation}
where $v^2 = (b_1-b_2)^T\S(b_1-b_2)$ and $W_t$ is a new one-dimensional
standard Brownian motion (it is obtained as $v^{-1}(b_1-b_2)^T \M W_t$ for
the old $N$-dimensional Brownian motion $W_t$ from
\eqref{sde}--\eqref{cov-B}). The coefficients $\theta_0$ and $\theta_1$ from
Theorem~\ref{th-many-agents} simplify to
\begin{align*}
&\theta_0 = (a-b_1)^T\M(b_1- b_2) + \frac12(b_1- b_2)^T (2\M-\S)(b_1- b_2),\\ 
&\theta_1 = (a-b_1)^T\M(b_1- b_2) + \frac12 (b_1- b_2)^T \S(b_1-b_2).
\end{align*}

\begin{theorem}
\label{th-two-agents}
Suppose $M=2$.

\noindent
I. If $v^2 > 0$, then the following statements hold true.
\begin{enumerate}[label=(I.\alph*)]
\item\label{th-two-agents-Ia} If $\theta_1> 0$, then agent 1 dominates.
\item\label{th-two-agents-Ib} If $\theta_0\ge 0$ and $\theta_1\le 0$, then both of the agents survive,
$\liminf_{t\to\infty} Y_t^1 = 0$,
$\limsup_{t\to\infty} Y_t^1 = 1$ a.s., the process $Y_t^1$ is recurrent
and has the invariant measure $F(dy) = \rho(y)dy$ with density
\[
\rho(y) = y^{\frac{2\theta_0}{v^2}-1}(1-y)^{-\frac{2\theta_1}{v^2}-1}, \qquad y\in(0,1).
\]
Moreover, if $\theta_0 \theta_1 < 0$, then $Y_t^1$ is positive recurrent,
$F([0,1]) = B(\frac{2\theta_0}{v^2},-\frac{2\theta_1}{v^2})< \infty$, and
$Y_t^1 \to F / F([0,1])$ in distribution as $t\to\infty$ ($B$ is the beta
function). If $\theta_0\theta_1=0$, then $Y_t^1$ is null recurrent and
$F([0,1])= \infty$.
\item\label{th-two-agents-Ic} If $\theta_0<0$, then agent 1 vanishes.
\end{enumerate}

\noindent
II. If $v^2 = 0$, then $Y_t^1$ is a non-random process and the following
statements hold true.
\begin{enumerate}[label=(II.\alph*)]
\item\label{th-two-agents-IIa} If $\theta_0>0$ and $\theta_1 \ge 0$, then agent 1 dominates.
\item\label{th-two-agents-IIb} If $\theta_0 > 0$ and $\theta_1< 0$, then both of the agents
survive and $\lim_{t\to\infty} Y_t^1 = \frac{\theta_0}{\theta_0-\theta_1}$.
\item\label{th-two-agents-IIc} If $\theta_0\le 0$ and $\theta_1 < 0$, then
agent 1 vanishes.
\item\label{th-two-agents-IId} If $\theta_0=\theta_1=0$ then $Y_t^1$ is
constant for all $t\ge0$.
\end{enumerate}
\end{theorem}

\begin{proof}
Claims \ref{th-two-agents-Ia} and \ref{th-two-agents-Ic} immediately follow
from Theorem~\ref{th-many-agents} (for \ref{th-two-agents-Ic} note that the
second agent has the corresponding coefficient $\tilde \theta_1 =
-\theta_0$). To prove \ref{th-two-agents-Ib}, let $Z_t =
\ln(Y_t^1/(1-Y_t^1))$. Define
\begin{equation}
f(z) = \theta_0 + \frac{\theta_1-\theta_0}{1+e^{-z}},\label{f}
\end{equation}
so that (cf.~\eqref{Zt})
\[
d Z_t = f(Z_t) dt + v dW_t.
\]
Let $s(z)$ and $m(dz)$ be the scale function and the speed measure of $Z_t$,
\[
\begin{aligned}
&s(z) = \int_0^z \exp\Bigl(-\int_0^y \frac{2f(u)}{v^2} du\Bigr) dy =
C\int_1^{e^z} (1+u)^{\frac{2(\theta_0-\theta_1)}{v^2}} u^{-1-\frac{2\theta_0}{v^2}} du,\\
&m(dz) = \frac{2}{v^2s'(z)} dz = \frac{2C}{v^2}
(1+e^z)^{\frac{2(\theta_0-\theta_1)}{v^2}} e^{-\frac{2\theta_0}{v^2}z},
\end{aligned}
\]
where $C=4^{\frac{\theta_1-\theta_0}{v^2}}$. In view of the conditions
$\theta_0\ge0$, $\theta_1\le 0$ we have $s(\pm\infty) :=
\lim_{z\to\pm\infty} s(z) = \pm\infty$, which implies that the process $Z_t$
is recurrent and its speed measure is the unique (up to multiplication by a
constant) invariant measure. If $\theta_0>0$ an $\theta_1<0$, we have
$m(\R)<\infty$, and then the process is positive recurrent and ergodic, so
$\lim_{t\to\infty}Z_t^1 = m/m(\R)$ in distribution, implying the claimed
result for $Y_t^1$. If $d_0d_1=0$, then $m(\R)=\infty$ and $Z_t$ is null
recurrent, so $Y_t^1$ is also null recurrent.

If $v=0$, then $Z_t$ has no Brownian part and claims
\ref{th-two-agents-IIa}--\ref{th-two-agents-IId} easily follow from analysis
of the solution of the corresponding ODE.
\end{proof}

\begin{corollary}
In the general model ($M\ge 2$), the strategy $\hat b=a$ is the unique
strategy which guarantees survival of an agent using it in any strategy
profile with any (positive) initial wealth.
\end{corollary}
\begin{proof}
A strategy of agent 1 surviving in any strategy profile must also survive
when all the other agents use the strategies $b_m=a$. In this case those
agents can be considered as a single agent, and then $\theta_0 = -\frac12
(b_1-a)^T\S(b_1-a) \le0$. By \ref{th-two-agents-Ic} and
\ref{th-two-agents-IIc} of Theorem~\ref{th-two-agents}, survival is possible
only when $\theta_0=0$, which implies $v^2=0$ and $\theta_1 =
-(a-b_1)^T\M(a-b_1)$. By \ref{th-two-agents-IIc},~\ref{th-two-agents-IId},
for survival it must hold that $\theta_1=0$, hence $b_1=a$.
\end{proof}

\begin{remark}
It is worth mentioning that the process $Y_t^1$ satisfies the stochastic
replicator equation of \cite{FudenbergHarris92} (see also
\cite{TaylorJonker78} for the seminal work on the deterministic replicator
equation, and \cite{FosterYoung90} for another form of the stochastic
equation).

Recall that the corresponding model can be formulated as follows. Consider a
symmetric two-player game with two pure strategies and a payoff matrix
$A=(A_{ij}) \in \R^{2\times2}$. There are two continuum populations of
players who are programmed to use, respectively, strategies 1 or 2 (e.g.\
strategies are phenotypes of biological species). Suppose the players are
randomly matched against each other. Let $S_t = (S_t^1,S_t^2)$ denote the
size of the populations $i=1,2$. Then the average payoff of a player from
population $i$ in a game against a random adversary is $(A S_t)_i$. The
model states that the population growth rates satisfy the equation
\[
\frac{d S_t^i}{S_t^i} =  (AS_t)_idt + \sigma_i d W_t^i,
\]
where $W_t^i$ are independent standard Brownian motions. Let $Y_t^i =
S_t^i/(S_t^1 + S_t^2)$ denote the proportion of players of type $i$. By
Ito's formula, we have
\begin{equation}
d Y_t^1 = Y_t^1 Y_t^2\bigl\{ ((-a_{22} +a_{12} + \sigma_2^2) + (a_{11} -
a_{21} - \sigma_1^2 + a_{22} - a_{12} - \sigma_2^2) Y_t^1) dt + v d
W_t\bigr\},\label{replicator}
\end{equation}
where $v=\sqrt{\sigma_1^2 + \sigma_2^2}$, and $W_t = v^{-1}(\sigma_1 W_t^1 +
\sigma_2 W_t^2)$ is a new Brownian motion. If $\sigma_1=\sigma_2=0$, one
gets the non-random replicator equation of \cite{TaylorJonker78}.

It is straightforward to check that equation \eqref{sde-2} is a particular
case of \eqref{replicator} with
\[
A =
\begin{pmatrix}
\theta_1 & 0 \\
0 & -\theta_0
\end{pmatrix},
\qquad
\sigma_1^2 = \sigma_2^2= \frac12(b_1-b_2)^T\S(b_1-b_2).
\]
Note that \eqref{replicator} admits one more type of asymptotic behavior,
which does not appear in our model because $\theta_0\ge \theta_1$, namely
when $\P(Y_t^1\to 1)>0$ and $\P(Y_t^1\to0) > 0$, see Proposition~1 of
\cite{FudenbergHarris92}.
\end{remark}

\section{Examples}
\label{sec-examples}
As an illustration of the survival and dominance conditions in terms of the
coefficients $\theta_0,\theta_1$, consider the model with two agents and two
assets. Assume $\mu_1=\mu_2 = 1/2$. Then the strategies of agents 1 and 2
are given respectively by the vectors $(b_1,-b_1)$ and $(b_2,-b_2)$. Denote
the coefficient $a_1$ of asset 1 simply by $a$, so that the coefficient
$a_2=-a$. Let $\sigma^2=\Var X_1$. Then
\[
\S = 
\begin{pmatrix}
\phantom{-}4\sigma^2 & -4\sigma^2 \\
-4\sigma^2 & \phantom{-}4\sigma^2
\end{pmatrix}.
\]
Note that we have $0\le \sigma^2 \le 1/4$; the maximum variance corresponds
to the case of a complete market \eqref{complete}, the zero variance to the
case of non-random payoffs. The coefficients $\theta_0$, $\theta_1$, $v$
become
\[
\begin{aligned}
&\theta_0 = 4(a-b_1)(b_1-b_2) + 4(1-2\sigma^2) (b_1-b_2)^2,\\
&\theta_1 = 4(a-b_1)(b_1-b_2) + 8\sigma^2 (b_1-b_2)^2,\\
&v = 4\sigma|b_1-b_2|.
\end{aligned}
\]
The conditions $\theta_0\ge0$ and $\theta_1\le 0$ turn into
\begin{align*}
&\theta_0 \ge 0 \iff (1-2\sigma^2)|b_1-b_2| \ge (b_1-a) \sgn(b_1-b_2),\\
&\theta_1 \le 0 \iff 2\sigma^2|b_1-b_2| \le (b_1-a) \sgn(b_1-b_2).
\end{align*}
These inequalities define regions with linear boundaries with slope $1$,
$2\sigma^2/(2\sigma^2-1)$, and $(2\sigma^2-1)/(2\sigma^2)$. In
Figure~\ref{figure}, we depict them for $\sigma^2=1/4,\,1/8,\,0$ and $a=0$.
The green region is where agent 1 dominates, the red region is where agent 2
dominates, in the blue region both of the agents survive. The color of the
boundaries bears the same meaning.

Figure~\ref{figure-simulation} shows simulated paths of the process $Y_t^1$
in the case $\sigma^2=1/8$ for the three pairs of the agents' strategies
$(b_1,b_2)$: $(-1/4,1)$, $(-1/3,1)$, $(-1/2,1)$. In the first pair, agent 1
dominates ($\theta_1>0$) and the process $Y_t^1$ is transient. In the second
pair both of the agents survive, the process $Y_t^1$ is null recurrent
($\theta_0>0$, $\theta_1=0$). In the third pair also both of the agents
survive, but $Y_t^1$ is positive recurrent ($\theta_0>0$, $\theta_1<0$). The
same realization of the Brownian motion $W_t$ was used in the simulations.

\bigskip
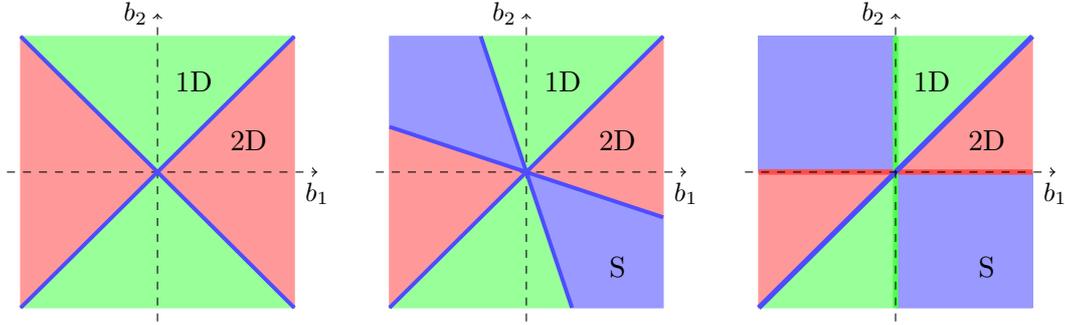
\begin{figure}[h]
\makebox[\textwidth][c]{%
\begin{tikzpicture}[scale=0.6]
\filldraw[green!40] (0,0)--(3,-3)--(-3,-3)--(0,0);
\filldraw[green!40] (0,0)--(3,3)--(-3,3)--(0,0);

\filldraw[red!40] (0,0)--(3,-3)--(3,3)--(0,0);
\filldraw[red!40] (0,0)--(-3,-3)--(-3,3)--(0,0);

\draw[dashed,->] (0,-3.3)--(0,3.5) node[left] {\small $b_2$};
\draw[dashed,->] (-3.3,0) node[below left,white!0] {\phantom{a}} --(3.5,0) node[below] {\small $b_1$};

\draw[line width=1.5pt,blue!70] (-3,3)--(3,-3);
\draw[line width=1.5pt,blue!70] (-3,-3)--(3,3);

\draw (0.8,2) node {1D};
\draw (2,0.7) node {2D};
\end{tikzpicture}
\hfill
\begin{tikzpicture}[scale=0.6]
\filldraw[blue!40] (0,0)--(3,-1)--(3,-3)--(1,-3)--(0,0);
\filldraw[blue!40] (0,0)--(-1,3)--(-3,3)--(-3,1)--(0,0);

\filldraw[green!40] (0,0)--(1,-3)--(-3,-3)--(0,0);
\filldraw[green!40] (0,0)--(-1,3)--(3,3)--(0,0);

\filldraw[red!40] (0,0)--(3,3)--(3,-1)--(0,0);
\filldraw[red!40] (0,0)--(-3,-3)--(-3,1)--(0,0);

\draw[dashed,->] (0,-3.3)--(0,3.5) node[left] {\small $b_2$};
\draw[dashed,->] (-3.3,0) --(3.5,0) node[below] {\small $b_1$};

\draw[line width=1.5pt,blue!70] (-3,1)--(3,-1);
\draw[line width=1.5pt,blue!70] (-1,3)--(1,-3);
\draw[line width=1.5pt,blue!70] (-3,-3)--(3,3);

\draw (0.8,2) node {1D};
\draw (2,-2.1) node {S};
\draw (2,0.7) node {2D};
\end{tikzpicture}
\hfill
\begin{tikzpicture}[scale=0.6]
\filldraw[blue!40] (0,0)--(3,0)--(3,-3)--(0,-3)--(0,0);
\filldraw[blue!40] (0,0)--(0,3)--(-3,3)--(-3,0)--(0,0);

\filldraw[green!40] (0,0)--(0,-3)--(-3,-3)--(0,0);
\filldraw[green!40] (0,0)--(3,3)--(0,3)--(0,0);

\filldraw[red!40] (0,0)--(3,0)--(3,3)--(0,0);
\filldraw[red!40] (0,0)--(-3,-3)--(-3,0)--(0,0);

\draw[line width=2pt,red!70] (-3,0)--(3,0);
\draw[line width=2pt,green!70] (0,-3)--(0,3);
\draw[line width=2pt,blue!70] (-3,-3)--(3,3);

\draw[dashed,->] (0,-3.3)--(0,3.5) node[left] {\small $b_2$};
\draw[dashed,->] (-3.3,0) --(3.5,0) node[below] {\small $b_1$};

\draw (0.8,2) node {1D};
\draw (2,-2.1) node {S};
\draw (2,0.7) node {2D};
\end{tikzpicture}
}
\caption{Regions of survival and dominance in the model with two agents and
two assets. Left: $\sigma^2=1/4$, middle: $\sigma^2=1/8$, right:
$\sigma=0$. ``1D'' (green region) -- agent 1 dominates; ``2D'' (red region) --
agent 2 dominates; ``S'' (blue region)
-- both survive.\\}
\centering
\label{figure}
\end{figure}

\bigskip
\begin{figure}[h]
\makebox[\textwidth][c]{%
\begin{tikzpicture}[xscale=0.07,yscale=2.4]
\draw[thick] (0.00,0.500)--(0.25,0.670)--(0.50,0.585)--(0.75,0.733)--(1.00,0.620)--(1.25,0.815)--(1.50,0.727)--(1.75,0.868)--(2.00,0.865)--(2.25,0.857)--(2.50,0.858)--(2.75,0.909)--(3.00,0.922)--(3.25,0.882)--(3.50,0.928)--(3.75,0.930)--(4.00,0.954)--(4.25,0.980)--(4.50,0.984)--(4.75,0.973)--(5.00,0.957)--(5.25,0.981)--(5.50,0.983)--(5.75,0.993)--(6.00,0.987)--(6.25,0.993)--(6.50,0.991)--(6.75,0.987)--(7.00,0.985)--(7.25,0.982)--(7.50,0.974)--(7.75,0.976)--(8.00,0.982)--(8.25,0.991)--(8.50,0.988)--(8.75,0.987)--(9.00,0.990)--(9.25,0.988)--(9.50,0.984)--(9.75,0.974)--(10.00,0.944)--(10.25,0.958)--(10.50,0.961)--(10.75,0.926)--(11.00,0.969)--(11.25,0.944)--(11.50,0.959)--(11.75,0.959)--(12.00,0.956)--(12.25,0.958)--(12.50,0.954)--(12.75,0.942)--(13.00,0.940)--(13.25,0.909)--(13.50,0.944)--(13.75,0.974)--(14.00,0.978)--(14.25,0.991)--(14.50,0.992)--(14.75,0.996)--(15.00,0.997)--(15.25,0.997)--(15.50,0.999)--(15.75,0.999)--(16.00,0.999)--(16.25,0.999)--(16.50,0.999)--(16.75,0.998)--(17.00,0.999)--(17.25,0.999)--(17.50,0.999)--(17.75,0.999)--(18.00,0.999)--(18.25,0.999)--(18.50,0.999)--(18.75,0.999)--(19.00,1.000)--(19.25,1.000)--(19.50,1.000)--(19.75,1.000)--(20.00,1.000)--(20.25,1.000)--(20.50,1.000)--(20.75,1.000)--(21.00,1.000)--(21.25,1.000)--(21.50,1.000)--(21.75,1.000)--(22.00,1.000)--(22.25,1.000)--(22.50,1.000)--(22.75,1.000)--(23.00,1.000)--(23.25,1.000)--(23.50,1.000)--(23.75,1.000)--(24.00,1.000)--(24.25,1.000)--(24.50,1.000)--(24.75,1.000)--(25.00,1.000)--(25.25,1.000)--(25.50,1.000)--(25.75,1.000)--(26.00,1.000)--(26.25,1.000)--(26.50,1.000)--(26.75,1.000)--(27.00,1.000)--(27.25,1.000)--(27.50,1.000)--(27.75,1.000)--(28.00,1.000)--(28.25,1.000)--(28.50,1.000)--(28.75,1.000)--(29.00,1.000)--(29.25,1.000)--(29.50,1.000)--(29.75,1.000)--(30.00,1.000)--(30.25,1.000)--(30.50,1.000)--(30.75,1.000)--(31.00,1.000)--(31.25,1.000)--(31.50,1.000)--(31.75,1.000)--(32.00,1.000)--(32.25,1.000)--(32.50,1.000)--(32.75,1.000)--(33.00,1.000)--(33.25,1.000)--(33.50,1.000)--(33.75,1.000)--(34.00,1.000)--(34.25,1.000)--(34.50,1.000)--(34.75,1.000)--(35.00,1.000)--(35.25,1.000)--(35.50,1.000)--(35.75,1.000)--(36.00,1.000)--(36.25,1.000)--(36.50,1.000)--(36.75,1.000)--(37.00,1.000)--(37.25,1.000)--(37.50,1.000)--(37.75,1.000)--(38.00,1.000)--(38.25,1.000)--(38.50,1.000)--(38.75,1.000)--(39.00,1.000)--(39.25,1.000)--(39.50,1.000)--(39.75,1.000)--(40.00,1.000)--(40.25,1.000)--(40.50,1.000)--(40.75,1.000)--(41.00,1.000)--(41.25,1.000)--(41.50,1.000)--(41.75,1.000)--(42.00,1.000)--(42.25,1.000)--(42.50,1.000)--(42.75,1.000)--(43.00,1.000)--(43.25,1.000)--(43.50,1.000)--(43.75,1.000)--(44.00,1.000)--(44.25,1.000)--(44.50,1.000)--(44.75,1.000)--(45.00,1.000)--(45.25,1.000)--(45.50,1.000)--(45.75,1.000)--(46.00,1.000)--(46.25,1.000)--(46.50,1.000)--(46.75,1.000)--(47.00,1.000)--(47.25,1.000)--(47.50,1.000)--(47.75,1.000)--(48.00,1.000)--(48.25,1.000)--(48.50,1.000)--(48.75,1.000)--(49.00,1.000)--(49.25,1.000)--(49.50,1.000)--(49.75,1.000)--(50.00,1.000);;
\draw[thick,->] (0,0) -- (55,0) node[below] {\small $t$};
\draw[thick,->] (0,0)--(0,1.2) node[left] {\small $Y$};
\draw (1.25,1)--(-1.25,1) node[left] {\scriptsize $1$};
\foreach \i in {0,10,20,30,40,50} {\draw (\i,0.03)--(\i,-0.03) node[below] {\scriptsize \i};}
\end{tikzpicture}
\hfill
\begin{tikzpicture}[xscale=0.07,yscale=2.4]
\draw[thick] (0.00,0.500)--(0.25,0.667)--(0.50,0.563)--(0.75,0.714)--(1.00,0.579)--(1.25,0.791)--(1.50,0.680)--(1.75,0.842)--(2.00,0.831)--(2.25,0.815)--(2.50,0.810)--(2.75,0.875)--(3.00,0.889)--(3.25,0.824)--(3.50,0.891)--(3.75,0.888)--(4.00,0.924)--(4.25,0.967)--(4.50,0.972)--(4.75,0.948)--(5.00,0.912)--(5.25,0.960)--(5.50,0.962)--(5.75,0.985)--(6.00,0.968)--(6.25,0.981)--(6.50,0.973)--(6.75,0.959)--(7.00,0.952)--(7.25,0.938)--(7.50,0.909)--(7.75,0.913)--(8.00,0.933)--(8.25,0.965)--(8.50,0.952)--(8.75,0.946)--(9.00,0.954)--(9.25,0.947)--(9.50,0.922)--(9.75,0.876)--(10.00,0.759)--(10.25,0.825)--(10.50,0.842)--(10.75,0.731)--(11.00,0.882)--(11.25,0.798)--(11.50,0.850)--(11.75,0.851)--(12.00,0.846)--(12.25,0.854)--(12.50,0.842)--(12.75,0.808)--(13.00,0.807)--(13.25,0.734)--(13.50,0.835)--(13.75,0.923)--(14.00,0.932)--(14.25,0.971)--(14.50,0.973)--(14.75,0.987)--(15.00,0.990)--(15.25,0.989)--(15.50,0.995)--(15.75,0.995)--(16.00,0.994)--(16.25,0.997)--(16.50,0.993)--(16.75,0.991)--(17.00,0.993)--(17.25,0.990)--(17.50,0.991)--(17.75,0.992)--(18.00,0.996)--(18.25,0.992)--(18.50,0.993)--(18.75,0.991)--(19.00,0.995)--(19.25,0.998)--(19.50,0.999)--(19.75,0.998)--(20.00,0.998)--(20.25,0.998)--(20.50,0.998)--(20.75,0.996)--(21.00,0.996)--(21.25,0.997)--(21.50,0.998)--(21.75,0.997)--(22.00,0.998)--(22.25,0.998)--(22.50,0.997)--(22.75,0.998)--(23.00,0.999)--(23.25,0.999)--(23.50,0.999)--(23.75,0.999)--(24.00,0.999)--(24.25,0.998)--(24.50,0.999)--(24.75,0.999)--(25.00,0.999)--(25.25,1.000)--(25.50,1.000)--(25.75,1.000)--(26.00,0.999)--(26.25,0.999)--(26.50,1.000)--(26.75,0.999)--(27.00,0.999)--(27.25,0.999)--(27.50,0.999)--(27.75,1.000)--(28.00,0.999)--(28.25,0.999)--(28.50,0.999)--(28.75,0.999)--(29.00,0.999)--(29.25,0.999)--(29.50,0.999)--(29.75,0.998)--(30.00,0.999)--(30.25,0.999)--(30.50,0.999)--(30.75,0.999)--(31.00,0.999)--(31.25,0.998)--(31.50,0.999)--(31.75,0.998)--(32.00,0.999)--(32.25,0.998)--(32.50,0.996)--(32.75,0.993)--(33.00,0.992)--(33.25,0.990)--(33.50,0.989)--(33.75,0.973)--(34.00,0.952)--(34.25,0.932)--(34.50,0.916)--(34.75,0.961)--(35.00,0.887)--(35.25,0.922)--(35.50,0.846)--(35.75,0.884)--(36.00,0.840)--(36.25,0.919)--(36.50,0.898)--(36.75,0.916)--(37.00,0.908)--(37.25,0.868)--(37.50,0.859)--(37.75,0.932)--(38.00,0.918)--(38.25,0.940)--(38.50,0.960)--(38.75,0.944)--(39.00,0.933)--(39.25,0.910)--(39.50,0.915)--(39.75,0.949)--(40.00,0.921)--(40.25,0.943)--(40.50,0.941)--(40.75,0.904)--(41.00,0.910)--(41.25,0.914)--(41.50,0.927)--(41.75,0.956)--(42.00,0.951)--(42.25,0.980)--(42.50,0.980)--(42.75,0.993)--(43.00,0.997)--(43.25,0.998)--(43.50,0.999)--(43.75,0.999)--(44.00,0.999)--(44.25,0.999)--(44.50,1.000)--(44.75,1.000)--(45.00,1.000)--(45.25,1.000)--(45.50,1.000)--(45.75,1.000)--(46.00,1.000)--(46.25,1.000)--(46.50,1.000)--(46.75,1.000)--(47.00,1.000)--(47.25,1.000)--(47.50,1.000)--(47.75,1.000)--(48.00,1.000)--(48.25,1.000)--(48.50,1.000)--(48.75,1.000)--(49.00,1.000)--(49.25,1.000)--(49.50,1.000)--(49.75,1.000)--(50.00,1.000);
\draw[thick,->] (0,0) -- (55,0) node[below] {\small $t$};
\draw[thick,->] (0,0)--(0,1.2) node[left] {\small $Y$};
\draw (1.25,1)--(-1.25,1) node[left] {\scriptsize $1$};
\foreach \i in {0,10,20,30,40,50} {\draw (\i,0.03)--(\i,-0.03) node[below] {\scriptsize \i};}
\end{tikzpicture}
\hfill
\begin{tikzpicture}[xscale=0.07,yscale=2.4]
\draw[thick] (0.00,0.500)--(0.25,0.656)--(0.50,0.513)--(0.75,0.669)--(1.00,0.488)--(1.25,0.734)--(1.50,0.573)--(1.75,0.776)--(2.00,0.744)--(2.25,0.709)--(2.50,0.695)--(2.75,0.788)--(3.00,0.802)--(3.25,0.678)--(3.50,0.789)--(3.75,0.772)--(4.00,0.836)--(4.25,0.924)--(4.50,0.928)--(4.75,0.854)--(5.00,0.759)--(5.25,0.885)--(5.50,0.882)--(5.75,0.952)--(6.00,0.880)--(6.25,0.924)--(6.50,0.880)--(6.75,0.812)--(7.00,0.779)--(7.25,0.732)--(7.50,0.647)--(7.75,0.682)--(8.00,0.758)--(8.25,0.866)--(8.50,0.808)--(8.75,0.786)--(9.00,0.809)--(9.25,0.783)--(9.50,0.702)--(9.75,0.600)--(10.00,0.421)--(10.25,0.582)--(10.50,0.634)--(10.75,0.472)--(11.00,0.747)--(11.25,0.594)--(11.50,0.692)--(11.75,0.693)--(12.00,0.689)--(12.25,0.710)--(12.50,0.683)--(12.75,0.629)--(13.00,0.634)--(13.25,0.533)--(13.50,0.702)--(13.75,0.851)--(14.00,0.860)--(14.25,0.937)--(14.50,0.934)--(14.75,0.966)--(15.00,0.972)--(15.25,0.963)--(15.50,0.981)--(15.75,0.980)--(16.00,0.966)--(16.25,0.982)--(16.50,0.947)--(16.75,0.922)--(17.00,0.935)--(17.25,0.897)--(17.50,0.901)--(17.75,0.906)--(18.00,0.941)--(18.25,0.878)--(18.50,0.896)--(18.75,0.847)--(19.00,0.917)--(19.25,0.954)--(19.50,0.980)--(19.75,0.949)--(20.00,0.956)--(20.25,0.941)--(20.50,0.925)--(20.75,0.835)--(21.00,0.836)--(21.25,0.866)--(21.50,0.904)--(21.75,0.863)--(22.00,0.891)--(22.25,0.891)--(22.50,0.852)--(22.75,0.907)--(23.00,0.927)--(23.25,0.953)--(23.50,0.937)--(23.75,0.921)--(24.00,0.879)--(24.25,0.863)--(24.50,0.944)--(24.75,0.936)--(25.00,0.918)--(25.25,0.947)--(25.50,0.933)--(25.75,0.955)--(26.00,0.898)--(26.25,0.905)--(26.50,0.924)--(26.75,0.779)--(27.00,0.859)--(27.25,0.818)--(27.50,0.851)--(27.75,0.898)--(28.00,0.856)--(28.25,0.803)--(28.50,0.710)--(28.75,0.782)--(29.00,0.847)--(29.25,0.905)--(29.50,0.816)--(29.75,0.736)--(30.00,0.779)--(30.25,0.857)--(30.50,0.878)--(30.75,0.848)--(31.00,0.773)--(31.25,0.720)--(31.50,0.821)--(31.75,0.758)--(32.00,0.857)--(32.25,0.814)--(32.50,0.635)--(32.75,0.540)--(33.00,0.570)--(33.25,0.560)--(33.50,0.608)--(33.75,0.436)--(34.00,0.387)--(34.25,0.395)--(34.50,0.410)--(34.75,0.693)--(35.00,0.428)--(35.25,0.608)--(35.50,0.458)--(35.75,0.602)--(36.00,0.533)--(36.25,0.746)--(36.50,0.693)--(36.75,0.747)--(37.00,0.728)--(37.25,0.644)--(37.50,0.651)--(37.75,0.823)--(38.00,0.784)--(38.25,0.836)--(38.50,0.883)--(38.75,0.829)--(39.00,0.791)--(39.25,0.726)--(39.50,0.749)--(39.75,0.838)--(40.00,0.754)--(40.25,0.813)--(40.50,0.803)--(40.75,0.708)--(41.00,0.723)--(41.25,0.734)--(41.50,0.767)--(41.75,0.856)--(42.00,0.835)--(42.25,0.927)--(42.50,0.922)--(42.75,0.972)--(43.00,0.987)--(43.25,0.990)--(43.50,0.995)--(43.75,0.996)--(44.00,0.993)--(44.25,0.993)--(44.50,0.995)--(44.75,0.995)--(45.00,0.998)--(45.25,0.998)--(45.50,0.999)--(45.75,0.998)--(46.00,0.997)--(46.25,0.998)--(46.50,0.996)--(46.75,0.996)--(47.00,0.997)--(47.25,0.996)--(47.50,0.981)--(47.75,0.990)--(48.00,0.993)--(48.25,0.990)--(48.50,0.975)--(48.75,0.975)--(49.00,0.958)--(49.25,0.950)--(49.50,0.917)--(49.75,0.912)--(50.00,0.910);
\draw[thick,->] (0,0) -- (55,0) node[below] {\small $t$};
\draw[thick,->] (0,0)--(0,1.2) node[left] {\small $Y$};
\draw (1.25,1)--(-1.25,1) node[left] {\scriptsize $1$};
\foreach \i in {0,10,20,30,40,50} {\draw (\i,0.03)--(\i,-0.03) node[below] {\scriptsize \i};}
\end{tikzpicture}
}
\caption{Simulations of the wealth process of agent 1 when $\sigma^2=1/8$ and the
initial wealth $Y_0^1=1/2$.
Left: $b_1=-1/4,b_2=1$ (transience), middle: $b_1=-1/3,b_2=1$ (null
recurrence), right: $b_1=-1/2$, $b_2=1$ (positive recurrence).}
\centering
\label{figure-simulation}
\end{figure}
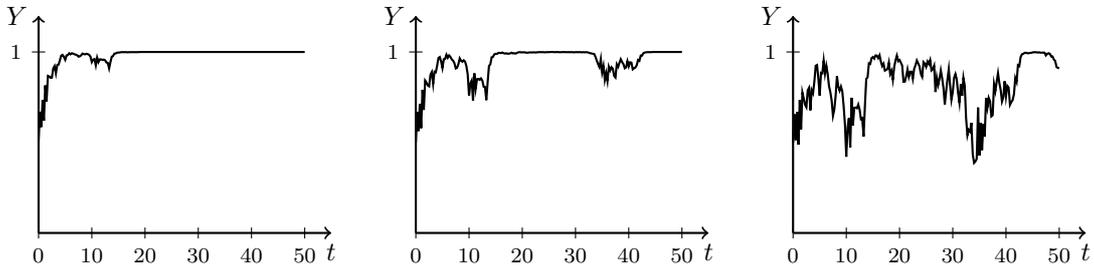

\section{Extension: the two-agent model with regime switching}
\label{sec-switching}
In this section we consider an extension of the two-agent model from
Section~\ref{sec-two-agents} in which the distribution of the payoffs
changes between two regimes at random moments of time.

Let $Q_t$ be a continuous-time ergodic Markov chain with two states and the
generator matrix
\[G =
\begin{pmatrix}
g_{11} & g_{12}\\ g_{21} & g_{22}
\end{pmatrix},
\]
where $g_{12} = -g_{11} >0$ and $g_{21}=-g_{22} > 0$. Let
$\pi = (\pi_1,\pi_2)$ denote its invariant distribution given by
\[
\pi_1 = \dfrac{g_{21}}{g_{12}+g_{21}},\quad \pi_2=
\dfrac{g_{12}}{g_{12}+g_{21}}.
\]
Suppose that when $Q_t$ is in state $i\in\{1,2\}$, the model is governed by
equations \eqref{sde}--\eqref{cov-B} with parameters $\mu(i),a(i) \in \R^N$,
$\sigma(i) \in \R^{N\times N}$. The agents' strategies $b_1,b_2$ are
constant. Then the wealth of agent 1 satisfies the SDE (cf.~\eqref{sde-2})
\begin{multline*}
dY_t^1 = Y_t^1(1-Y_t^1)\Bigl\{\bigl((a(Q_t)- b_2)^T\M(Q_t)(b_1- b_2) \\-Y_t^1(b_1-
b_2)^T\M(Q_t)(b_1- b_2)\bigr) dt + v(Q_t) d W_t\Bigr\},
\end{multline*}
where $\M(Q_t),\S(Q_t)$ are defined as in \eqref{matrices}, and $v(Q_t) =
\sqrt{(b_1-b_2)^T\S(Q_t)(b_1-b_2)}$.

Let $\theta_0(i)$, $\theta_1(i)$ denote the coefficients from
\eqref{d0}--\eqref{d1} corresponding to state $i$. Define
\[
\bar \theta_j = \E^\pi (\theta_j(Q_t)) := \sum_{i=1}^2 \theta_j(i) \pi_i, \qquad j=1,2,
\]
where $\E^\pi$ is the expectation with respect to the stationary
distribution of $Q_t$. Note that $\bar \theta_0 \ge \bar \theta_1$ by
Theorem~\ref{th-many-agents}.

\begin{theorem}
\label{th-switching}
Assume that $v(i)>0$ for $i=1,2$.
\begin{enumerate}[leftmargin=*,label=(\alph*)]
\item\label{th-switching-a} If $\bar \theta_1 > 0$, then agent 1 dominates.
\item\label{th-switching-b} if $\bar \theta_0>0$ and $\bar \theta_1<0$, then both of the agents survive, the
process $Y_t^1$ is positive recurrent, and $\limsup_{t\to\infty} Y_t^1 = 1$,
$\liminf_{t\to\infty} Y_t^1 = 0$.
\item\label{th-switching-c} If $\bar \theta_0 < 0$, then agent 1 vanishes.
\end{enumerate}
\end{theorem}

\begin{proof}
Let $Z_t = \ln(Y_t^1/(1-Y_t^1))$. Similarly to \eqref{f}, define
\[
f(z,i) = \theta_0(i) +   \frac{\theta_1(i)-\theta_0(i)}{1+e^{-z}}.
\]
Then
\[
d Z_t = f(Z_t,Q_t)dt  + v(Q_t) d W_t.
\]
Note that always $\theta_1(i) \le f(z,i) \le \theta_0(i)$. In case
\ref{th-switching-a}, using the ergodicity of $Q_t$ and the strong law of
large numbers applied to the martingale $\int_0^t v(Q_s) d W_s$, we find
\[
\lim_{t\to\infty} \frac{Z_t}{t} \ge \E^\pi \theta_1(Q_t) = \bar \theta_1>0,
\]
hence $Z_t\to\infty$ and $Y_t^1 \to 1$ as $t\to\infty$. Similarly, in
case~\ref{th-switching-c} we have $Z_t\to-\infty$.

Consider case \ref{th-switching-b}. In order to prove the positive
recurrence it is enough to prove it for a set of the form $I\times\{1,2\}$,
where $I$ is an interval $(z_1,z_2) \subset \R$ (see, e.g.,
\citet[Sec.~3.2]{YinZhu09}). Since $\bar \theta_0>0$ and $\bar \theta_1 <
0$, we can find $z_1\le z_2$ such that
\[
\E ^\pi f(z_1, Q_t) >0, \qquad \E ^\pi f(z_2, Q_t) <0.
\]

Consider an initial condition $z_0 \le z_1$ and let $\tau = \inf\{t\ge 0:
Z_t = z_1\}$. We are going to find a function $V(z,i)$ such that $V(z,i)\ge
0$ and $L V(z,i) \le -u$ for some $u>0$ and all $z\le z_1$, $i=1,2$, where
$L$ is the infinitesimal generator of $(Z,Q)$:
\[
L V(z,i) = \frac{v^2(i)}{2} V''(z,i) + f(z,i) V'(z,i) + V(z,1) g_{i1} +
V(z,2) g_{i2}.
\]
Observe that it is possible to find $\epsilon>0$ such that for $c =
\frac{f(z_1,2)}{2g_{21}} - \epsilon $ we have
\begin{align}
&f(z_1,1) + 2cg_{12} > 0,\label{f-1}\\
&f(z_1,2) - 2c g_{21} > 0.\label{f-2}
\end{align}
Indeed, for $\epsilon=0$ we have the equality to zero in \eqref{f-2} and the
strict inequality in \eqref{f-1}, which follows from that $\E^\pi f(z_1, Q)
= (g_{12} + g_{21})^{-1} (f(z_1,1)g_{21} + f(z_1,2)g_{12}) > 0$. Hence such
$\epsilon$ an be found by continuity. Since the function $f(z,i)$ is
non-increasing in $z$, inequalities \eqref{f-1}--\eqref{f-2} hold for all
$z\le z_1$.

For $\gamma \in (0,|c|^{-1})$, consider the function $V(z,i)$ defined by
\[
V(z,1) = (1+c\gamma) e^{-\gamma z}, \qquad V(z,2) = (1-c\gamma)e^{-\gamma z}.
\]
We have
\begin{align*}
&LV(z,1) = \frac{v^2(i)}{2} \gamma^2(1+c\gamma) - \gamma(f(z,1)
(1+c\gamma) + 2c g_{12})e^{-\gamma z},\\
&LV(z,2) = \frac{v^2(i)}{2} \gamma^2(1-c\gamma) - \gamma(f(z,1)
(1-c\gamma) -2c g_{21})e^{-\gamma z}.
\end{align*}
In view of \eqref{f-1}--\eqref{f-2}, taking $\gamma$ sufficiently small, it
is possible to make $LV(z,1) \le -u $ and $LV(z,2) \le -u$ for some $u>0$
and all $z\le z_1$. By Ito's formula, under the initial condition $Z_0=z_0$
and $Q_0=i_0$, we find
\[
\E V(Z_{\tau\wedge t}, Q_{\tau\wedge t}) = V(z_0, i_0)  + \E \int_0^{\tau\wedge t} L V(Z_s,Q_s) ds \le
V(z_0,i_0) -u \E (\tau\wedge t).
\]
By the monotone convergence theorem applied with $t\to\infty$, we have $\E
\tau \le u^{-1}V(z_0,i_0)$. Hence the set $I\times \{1,2\}$ can be reached
from a point $z_0\le z_1$ in time with finite expectation. In a similar way,
we consider points $z_0 \ge z_2$, and establish the positive recurrence of
the set $I\times\{1,2\}$.
\end{proof}

\begin{remark}
Asymptotic behavior of a solution of a replicator equation with regime
switching was studied by \cite{Vlasic15}, although his assumptions are more
complicated than ours (see Assumption 4.1 of that paper).
\end{remark}

\renewcommand{\thesection}{Appendix}
\section{Conditions for convergence to a diffusion process}
\label{appendix}
This appendix states a result about convergence in distribution of
discrete-time processes with uniformly bounded jumps to a diffusion process
in a form convenient for our needs.

Consider a stochastic differential equation
\begin{equation}
d Y_t = f(Y_t) dt + g(Y_t)dW_t, \qquad Y_0=y_0 \in \R^M,\label{sde-weak}
\end{equation}
where $f\colon \R^M\to\R^M$, $g\colon \R^M \to \R^{M\times M}$ are
measurable functions and $W_t$ is a Brownian motion in $\R^M$ with
covariance matrix $\sigma = (\sigma_{mk})\in\R^{M\times M}$, i.e.\ $\E W_t^m
W_t^k = \sigma_{mk}t$. Assume that equation \eqref{sde-weak} has a unique
weak solution and denote by $\Law(Y)$ the distribution which it generates on
the Skorokhod space $D(\R_+; \R^M)$ of \cadlag\ function $\alpha\colon \R_+
\to \R^M$ (the support of this distribution lies in the subspace of
continuous functions).

Let $Y_t^{\delta}$ be piecewise-constant \cadlag\ processes in $\R^M$ with
the same initial values $Y_0^\delta = y_0$, which are indexed by a parameter
$\delta>0$. Assume that $Y_t^\delta$ is constant on intervals
$[i{\delta},(i+1){\delta})$, $i\ge 0$, with jumps $\|\Delta \tilde
Y_i^\delta\| \le c$ for some constant $c$ which is the same for all
$\delta$. Let $\F_t^\delta = \sigma(Y_s^\delta; s\le t)$ be the natural
filtration of $Y_t^\delta$. We are interested in conditions for the weak
convergence
\[
\Law(Y^\delta) \to \Law(Y), \qquad \delta\to0.
\]

To formulate these conditions, introduce the predictable processes
$B_t^\delta,C_t^\delta$ with values in $\R^M$ and $\R^{M\times M}$,
respectively, defined by
\begin{align}
\label{B-delta}
&B^{\delta,m}_t = \sum_{1\le i \le \floor{t/\delta}}
  \E(\Delta Y_{i\delta}^{\delta,m} \mid \F_{(i-1)\delta}^\delta), \\
\label{C-delta}
&C^{\delta,mk}_t = \sum_{1\le i \le \floor{t/\delta}}\biggl(
  \E(\Delta Y_{i\delta}^{\delta,m} \Delta Y_{i\delta}^{\delta,k} \mid \F_{(i-1)\delta}^\delta)
  -\E(\Delta Y_{i\delta}^{\delta,m}\mid \F_{(i-1)\delta}^\delta)
   \E(\Delta Y_{i\delta}^{\delta,k}\mid \F_{(i-1)\delta}^\delta) \biggr)\notag.
\end{align}
These processes are the first and modified second predictable
characteristics without truncation of $Y_t^\delta$, see
\citet[IX.3.25]{JacodShiryaev02}. Also, on the space $D(\R_+; \R^M)$ define
the functionals
\begin{equation}
B_t(\alpha) = \int_0^t f(\alpha_s) ds, \qquad C_t(\alpha) = \int_0^t
g(\alpha_s)\sigma g(\alpha_s)^T\, ds, \qquad \alpha \in D(\R_+;
\R^M),\label{B-C}
\end{equation}
so that $B_t(Y)$ and $C_t(Y)$ are the first and second predictable
characteristics of $Y_t$.

The next proposition follows from Theorem~IX.3.27 of \cite{JacodShiryaev02}.
\begin{proposition}
\label{th-convergence}
Suppose the following conditions hold true:
\begin{enumerate}[leftmargin=*,itemsep=0mm,label=(\alph*),widest=a]
\item\label{conv-1} equation \eqref{sde-weak} has a unique weak solution;
\item\label{conv-2} the functions $f(y)$ and $g(y)$ are continuous;
\item\label{conv-3} there exists a non-random continuous increasing function
$F(t)$ such that for any $\alpha \in D(\R_+;\R^M)$ the function
$F(t) - \int_0^t \sum_m |f^m(\alpha_s)|ds   - \mathrm{tr}( C_t(\alpha))$ is
increasing;
\item\label{conv-4} $\sup\limits_{s\le t} \|B_s^\delta -
B_s(Y^\delta) \| \to 0$ in probability as $\delta\to 0$ for all $t > 0$;
\item\label{conv-5} $C_t^\delta - C_t(Y^\delta) \to 0$ in probability as
$\delta\to 0$ for all $t > 0$;
\item\label{conv-6} $\sum_{1\le i \le \floor{t/\delta}} \E(h(\Delta
Y_{i\delta}^\delta)\mid \F_{(i-1)\delta}^\delta) \to 0$ in probability for
all $t > 0$ and any continuous bounded function $h(y)\colon \R^M \to \R$
which vanishes in a neighborhood of zero.
\end{enumerate}
Then $\Law(Y^\delta)$ weakly converges to $\Law(Y)$ as $\delta \to 0$.
\end{proposition}

\small 
\setlength{\bibsep}{0.2em plus 0.3em}
\bibliographystyle{apalike}
\bibliography{approx}

\begin{thebibliography}{}

\bibitem[Alchian, 1950]{Alchian50}
Alchian, A.~A. (1950).
\newblock Uncertainty, evolution, and economic theory.
\newblock {\em Journal of Political Economy}, 58(3):211--221.

\bibitem[Amir et~al., 2005]{AmirEvstigneev+05}
Amir, R., Evstigneev, I.~V., Hens, T., and Schenk-Hopp{\'e}, K.~R. (2005).
\newblock Market selection and survival of investment strategies.
\newblock {\em Journal of Mathematical Economics}, 41(1-2):105--122.

\bibitem[Amir et~al., 2013]{AmirEvstigneev+13}
Amir, R., Evstigneev, I.~V., and Schenk-Hopp{\'e}, K.~R. (2013).
\newblock Asset market games of survival: a synthesis of evolutionary and
  dynamic games.
\newblock {\em Annals of Finance}, 9(2):121--144.

\bibitem[Blume and Easley, 1992]{BlumeEasley92}
Blume, L. and Easley, D. (1992).
\newblock Evolution and market behavior.
\newblock {\em Journal of Economic Theory}, 58(1):9--40.

\bibitem[Blume and Easley, 2006]{BlumeEasley06}
Blume, L. and Easley, D. (2006).
\newblock If you're so smart, why aren't you rich? {B}elief selection in
  complete and incomplete markets.
\newblock {\em Econometrica}, 74(4):929--966.

\bibitem[Bottazzi and Dindo, 2014]{BottazziDindo14}
Bottazzi, G. and Dindo, P. (2014).
\newblock Evolution and market behavior with endogenous investment rules.
\newblock {\em Journal of Economic Dynamics and Control}, 48:121--146.

\bibitem[Bottazzi and Giachini, 2017]{BottazziGiachini17}
Bottazzi, G. and Giachini, D. (2017).
\newblock Wealth and price distribution by diffusive approximation in a
  repeated prediction market.
\newblock {\em Physica A: Statistical Mechanics and its Applications},
  471:473--479.

\bibitem[Bottazzi and Giachini, 2019]{BottazziGiachini19}
Bottazzi, G. and Giachini, D. (2019).
\newblock Far from the madding crowd: Collective wisdom in prediction markets.
\newblock {\em Quantitative Finance}, 19(9):1461--1471.

\bibitem[Boyd and Vandenberghe, 2004]{BoydVandenberghe04}
Boyd, S. and Vandenberghe, L. (2004).
\newblock {\em Convex Optimization}.
\newblock Cambridge University Press.

\bibitem[De~Long et~al., 1990]{DeLongShleifer+90}
De~Long, J.~B., Shleifer, A., Summers, L.~H., and Waldmann, R.~J. (1990).
\newblock Noise trader risk in financial markets.
\newblock {\em Journal of Political Economy}, 98(4):703--738.

\bibitem[Drokin and Zhitlukhin, 2020]{DrokinZhitlukhin20}
Drokin, Y. and Zhitlukhin, M. (2020).
\newblock Relative growth optimal strategies in an asset market game.
\newblock {\em Annals of Finance}, 16:529--546.

\bibitem[Evstigneev et~al., 2002]{EvstigneevHens+02}
Evstigneev, I.~V., Hens, T., and Schenk-Hopp{\'e}, K.~R. (2002).
\newblock Market selection of financial trading strategies: Global stability.
\newblock {\em Mathematical Finance}, 12(4):329--339.

\bibitem[Evstigneev et~al., 2009]{EvstigneevHens+09}
Evstigneev, I.~V., Hens, T., and Schenk-Hopp{\'e}, K.~R. (2009).
\newblock Evolutionary finance.
\newblock In {\em Handbook of Financial Markets: Dynamics and Evolution},
  Handbooks in Finance, chapter~9, pages 507--566. Elsevier.

\bibitem[Evstigneev et~al., 2016]{EvstigneevHens+16}
Evstigneev, I.~V., Hens, T., and Schenk-Hopp{\'e}, K.~R. (2016).
\newblock Evolutionary behavioral finance.
\newblock In Haven, E. et~al., editors, {\em The Handbook of Post Crisis
  Financial Modelling}, pages 214--234. Palgrave Macmillan UK.

\bibitem[Foster and Young, 1990]{FosterYoung90}
Foster, D. and Young, P. (1990).
\newblock Stochastic evolutionary game dynamics.
\newblock {\em Theoretical Population Biology}, 38(2):219--232.

\bibitem[Fudenberg and Harris, 1992]{FudenbergHarris92}
Fudenberg, D. and Harris, C. (1992).
\newblock Evolutionary dynamics with aggregate shocks.
\newblock {\em Journal of Economic Theory}, 57(2):420--441.

\bibitem[Gikhman and Skorokhod, 1972]{GikhmanSkorokhod72}
Gikhman, I.~I. and Skorokhod, A.~V. (1972).
\newblock {\em Stochastic Differential Equations}.
\newblock Springer-Verlag, New York, Heidelberg.

\bibitem[Jacod and Shiryaev, 2002]{JacodShiryaev02}
Jacod, J. and Shiryaev, A. (2002).
\newblock {\em Limit Theorems for Stochastic Processes}.
\newblock Springer, Berlin, 2nd edition.

\bibitem[Kelly, 1956]{Kelly56}
Kelly, Jr, J.~L. (1956).
\newblock A new interpretation of information rate.
\newblock {\em Bell System Technical Journal}, 35(4):917--926.

\bibitem[Sandroni, 2000]{Sandroni00}
Sandroni, A. (2000).
\newblock Do markets favor agents able to make accurate predictions?
\newblock {\em Econometrica}, 68(6):1303--1341.

\bibitem[Taylor and Jonker, 1978]{TaylorJonker78}
Taylor, P.~D. and Jonker, L.~B. (1978).
\newblock Evolutionary stable strategies and game dynamics.
\newblock {\em Mathematical Biosciences}, 40(1-2):145--156.

\bibitem[Vlasic, 2015]{Vlasic15}
Vlasic, A. (2015).
\newblock Stochastic replicator dynamics subject to {M}arkovian switching.
\newblock {\em Journal of Mathematical Analysis and Applications},
  427(1):235--247.

\bibitem[Yin and Zhu, 2009]{YinZhu09}
Yin, G.~G. and Zhu, C. (2009).
\newblock {\em Hybrid switching diffusions: properties and applications},
  volume~63.
\newblock Springer Science \& Business Media.

\bibitem[Zhitlukhin, 2020]{Zhitlukhin20}
Zhitlukhin, M. (2020).
\newblock A continuous-time asset market game with short-lived assets.
\newblock {\em arXiv:2008.13230}.

\bibitem[Zhitlukhin, 2021a]{Zhitlukhin21b}
Zhitlukhin, M. (2021a).
\newblock Capital growth and survival strategies in a market with endogenous
  prices.
\newblock {\em arXiv:2101.09777}.

\bibitem[Zhitlukhin, 2021b]{Zhitlukhin21a}
Zhitlukhin, M. (2021b).
\newblock Survival investment strategies in a continuous-time market model with
  competition.
\newblock {\em International Journal of Theoretical and Applied Finance},
  24(01):2150001.

\end{thebibliography}

\end{document}